\documentclass[sigconf,final]{acmart}

\usepackage[utf8]{inputenc}
\usepackage[T1]{fontenc}
\usepackage{url}
\usepackage{booktabs} 

\usepackage[shortlabels, inline]{enumitem}

\usepackage{amsthm}

\usepackage{mathrsfs}
 \floatstyle{boxed}
 \newfloat{algo}{h}{loa}
\floatname{algo}{{\textnormal{Algo}}}

\usepackage[ruled,linesnumbered,vlined]{algorithm2e}

\usepackage{stmaryrd}

\DeclareMathOperator{\val}{val}
\DeclareMathOperator{\tr}{trace}

\newcommand{\openr}[1]{\ensuremath{[\![#1[\![} }
\newcommand{\wb}[1]{\ensuremath{\overline{#1}}}

\def\Q{\ensuremath{\mathbb{Q}}}
\def\Z{\ensuremath{\mathbb{Z}}}
\def\F{\ensuremath{\mathbb{F}}}
\def\R{\ensuremath{\mathbb{R}}}

\def\M{\ensuremath{\mathsf{M}}}

\def\abs{\ensuremath{\mathrm{abs}}}
\def\rel{\ensuremath{\mathrm{rel}}}

\newcommand{\N}{\mathbb N}
\newcommand{\Zp}{\Z_p}
\newcommand{\Qp}{\Q_p}
\newcommand{\Fp}{\mathbb{F}_p}
\newcommand{\OK}{\mathcal{O}_K}
\newcommand\Qt{\Q\llbracket T \rrbracket}
\newcommand\xs{x^\star}
\newcommand{\tran}{\!^t}

\newtheorem{lem}{Lemma}[section]
\newtheorem{prop}[lem]{Proposition}
\newtheorem{theo}[lem]{Theorem}
\newtheorem{cor}[lem]{Corollary}
\newtheorem{conj}[lem]{Conjecture}
\newtheorem{defn}[lem]{Definition}
\newtheorem{ass}[lem]{Assumption}
\theoremstyle{remark}
\newtheorem{rmk}[lem]{Remark}

\definecolor{answer}{rgb}{0,0.5,0.2}

\definecolor{todo}{rgb}{0,0.2,0.5}
\newcommand{\xav}[1]{\textcolor{todo}{{\bf Xav:} #1}}

\setcopyright{rightsretained}

\acmDOI{}

\acmISBN{}

\acmConference[ISSAC'20]{The 45th International Symposium on Symbolic and Algebraic Computation}{July 2020}{Kalamata, Greece} 
\acmYear{2020}
\copyrightyear{2020}

\acmPrice{15.00}

\begin{document}
\title{On A Non-Archimedean Broyden Method}

\author{Xavier Dahan}
\orcid{0001-6042-6132}
\affiliation{%
  \institution{Tohoku University,  IEHE}
  \city{Sendai} 
  \state{Japan} 
  \postcode{980-8576}
}
 \email{xdahan@gmail.com}

\author{
Tristan Vaccon}
  \affiliation{Universit\'e de Limoges;
  \institution{CNRS, XLIM UMR 7252}
  \city{Limoges, France}  
  \postcode{87060}  
  }
  \email{tristan.vaccon@unilim.fr}

\begin{abstract}
Newton's method is an ubiquitous tool to solve equations,
both in the archimedean and 
non-archimedean settings
--- for which it does not really differ. 
Broyden was the instigator of what is called
``quasi-Newton methods''. These methods
use an iteration step where one does not
need to compute a complete Jacobian
matrix nor its inverse.
We provide an adaptation of 
Broyden's method in a general non-archimedean setting,
compatible with the lack of inner product,
and study its Q and R convergence.
We prove that our adapted method
converges at least Q-linearly
and R-superlinearly with
R-order $2^{\frac{1}{2m}}$ in dimension $m.$
Numerical data are provided.
\end{abstract}

%
%



\maketitle

\section{Introduction}
\vspace{-1mm}
\paragraph{In the numerical world.} 
Quasi-Newton methods refer to a class of variants of Newton's method
for solving square nonlinear systems,
 with the twist that  the inverse of the Jacobian matrix is ``approximated''
by another matrix. When compared to Newton's method, they benefit from a cheaper update at each iteration (See {\em e.g.} \cite[p.49-50, 53]{DM77}), but suffer from a smaller rate of convergence.
They were mainly introduced by
 Broyden in~\cite{Broyden65},
which has sparked numerous
improvements, generalizations, and variants
(see the surveys~\cite{DM77, martinez2000}).
It is now a fundamental numerical tool (that finds
its way in entry level numerical analysis textbooks~\cite[\S~10.3]{BuFa9}).
To some extent, this success stems from: the specificities of machine precision arithmetic
as commonly used in the numerical community,
the fact that Newton's method is usually not quadratically
convergent from step one,
and that the arithmetic cost of an iteration is independent of the
quality of the approximation reached.
In another direction, variants of Broyden's method have known
dramatic success for unconstrained optimization ---
the target system is the gradient of the objective function, the zeros
are then critical points---
where it takes advantage  of the special structure of the Hessian (see Sec. 7 of \cite{DM77}).
Another  appealing feature  of Broyden's method is the possibility to
design derivative-free methods
generalizing to the multivariate case the classical secant method (which
can be thought of as Broyden's in dimension one). This feature is a main motivation for this work.

\vspace{-7pt}

\paragraph{Non-archimedean.}
It is a natural wish to
transpose such a fundamental numerical method to the non-archimedean
framework, offering new tools to perform {\em exact} computations,
typically for systems with $p$-adic or power series coefficients.
For this adaptation, several non-trivial difficulties have to be overcome: \textit{e.g.} no inner products, a more difficult proof of convergence, or a management
of arithmetic at finite precision far more subtle.
This article presents satisfactory solutions for all these difficulties,
which we believe can be expanded to a broader variety of
  quasi-Newton methods.


Bach proved in \cite{Bach09} that in dimension one, the secant method
can be on an equal footing with Newton's method in terms of complexity.
We investigate how this comparison is less engaging in superior dimension
(see Section~\ref{sec:precision_finie}).
To our opinion, this is due
to the remarkable behavior of  Newton's method in the non-archimedean setting.
No inversion of the Jacobian is required at each iteration
(simply a matrix multiplication,
this is now classical see~\cite{kung1974computing, brent1978fast, kung1978all}).
The evaluation of the Jacobian is also efficient for polynomial functions
(in dimension $m$, it involves only $O(m)$ evaluations, instead of $m^2$
over $\R$, see~\cite{baur1983complexity}).
It displays also true quadratic behavior from step one
which, when combined with the natural use of finite precision arithmetic (against
machine precision over $\R$),
offers a ratio cost/precision gained that is hard to match.

And indeed, our results show that for large dimension $m$ and polynomials as input,
there is little hope for Broyden to outperform Newton,
although it  depends on the order of superlinear convergence of Broyden's method.
In this respect more investigation is necessary, but for now the interest
  lies more in  the theoretical advances
  and in the  situations mentioned in   ``{\em Motivations}'' thereafter.

\vspace{-7pt}

\paragraph{Relaxed arithmetic} Since the cost of one iteration
of Broyden's method involves $m^2$ instead of $m^\omega$ for Newton,
we should mention the {\em relaxed} framework (a.k.a online~\cite{FiSt74}) which show essentially
the same decrease of complexity, while maintaining quadratic convergence.
It has been implemented efficiently for power series~\cite{vdH02relax},
and for $p$-adic numbers~\cite{berthomieu2011}.
In case of a smaller $m$ and a larger precision of approximation required,
 FFT
trading~\cite{van2010newton} has to be mentioned.
These techniques are however unlikely  to be suited to the Broyden iteration,
since it is  {\em a priori} not described by a fixed-point equation, 
a necessity for the relaxed machinery.

\vspace{-8pt}

\paragraph{Motivations}
As explains Remark~\ref{rmk:NvsB}, it seems unlikely in the non-archimedean world 
that with polynomials or rational fractions, a quasi-Newton method
meets the standard of Newton's method.
The practical motivations concern:

1/ Derivative-free method:
instead of starting with the Jacobian at precision one,  use  a divided-difference matrix.
A typical application is  when the function is given by a ``black-box''
  and there is no direct access to the Jacobian.
  
2/  When computing the  Jacobian does not allow shortcuts like in the case of rational fractions~\cite{baur1983complexity},
evaluating it may require up to  $ L m^2$ operations, where $L$ is the complexity of evaluation of the input function.
Regarding the complexity of Remark~\ref{rmk:NvsB},
Broyden's method then becomes   beneficial when $L \gtrsim m^2 - m^{\omega-1}  $.
  
  3/ 
  While Newton's method over general Banach spaces of infinite dimension can be made effective when the differential is effectively representable (integral equations~\cite[\S~5]{KelleySachs1990}\cite{kelley1985broyden}
  are a typical example),
  it is in general difficult or impossible to compute it.
  On the other hand, Broyden's method or its variants have the ability to work with approximations of the
  differential, including of {\em  finite rank}, by considering a projection
  (as shown in~\cite{KelleySachs1990,kelley1985broyden} and the references therein;
  the dimension of the projection is increased at each iteration).
  In the non-archimedean context, 
  ODEs with parameters, for example initial conditions, constitute a natural application.

\vspace{-5pt}

\paragraph{Organization of the paper} Definitions and notations are introduced in Section \ref{sec:notations}.
Section \ref{sec:adaptation} explains how Broyden's method can be 
adapted to an ultrametric setting.
In Section \ref{sec:convergence}, we study the Q and R-order of convergence of Broyden's method (see Definition \ref{defn:Q_R_CV}), presenting our main results.
It is followed by Section \ref{sec:Qsuperlinearity}, where are  introduced
developments and conjectures on Q-superlinearity. 
Finally, in Section \ref{sec:precision_finie}, we explain how our Broyden's method can
be implemented with dynamical handling of the precision, and we conclude
with some numerical data in Section \ref{sec:data}.

\vspace{-1mm}
\section{Broyden's method and notations}
\label{sec:notations}
\vspace{-1mm}
\subsection{General notations}

Throughout the paper, $K$ refers to a complete,
discrete valuation field, $\val : K \twoheadrightarrow \Z \cup \{+\infty\}$ to its valuation,
$\OK$ its ring of integers and $\pi$ a uniformizer.\footnote{Discrete valuation is only
needed in Section \ref{sec:precision_finie}. 
For the rest complete and ultrametric is enough.} For $k \in \N$, we write $O(\pi^k)$ for $\pi^k \OK.$

Let $m \in \Z_{\geq 1}$. We are interested in computing
an approximation of a non-singular zero $\xs$ of $f : K^m \rightarrow K^m$
through an iterative sequence of approximations, $(x_n)_{n \in \N} \in (K^m)^\N.$
Note that all our vectors are column-vectors.
For any $x \in K^m$ where it is well-defined, we denote by $f'(x) \in M_m(K)$ 
the Jacobian matrix of $f$ at $x.$ 
We will use the following notations (borrowed from \cite{Gay79}): 
\begin{equation}
f_n = f(x_n), \quad y_n = f_{n+1} - f_{n},
\quad 
s_n = x_{n+1} -x_{n}
\end{equation}
We denote by $(e_1,\dots, e_m)$ the canonical basis
of $K^m.$ In $K^m,$ $O(\pi^k)$
means $O(\pi^k)e_1+\dots+ O(\pi^k)e_m.$

Newton's iteration produces a sequence $(x_n)_{n \in \N}$ 
given by:
\vspace{-1mm}
\begin{equation}
    x_{n+1} = x_n - f'(x_n)^{-1} \cdot f(x_n). \tag{N} \label{eqn:Newton}
\end{equation}
\vspace{-1mm}
For quasi-Newton methods, the iteration is given by:
\begin{equation}
    x_{n+1} = x_n - B_n^{-1} \cdot f(x_n), 
\quad (\Rightarrow s_n = -B_n^{-1} \cdot f_n)  \tag{QN} \label{eqn:QNewton}
\end{equation}
with $B_n$ presumably not far from $f'(x_n).$
More precisely, it is a generalization
of the design of the secant method over $K$
where one approximates $f'(x_n)$ by $\frac{f(x_n)-f(x_{n-1})}{x_n-x_{n-1}}.$
In quasi-Newton, it is thus required that:
\vspace{-1mm}
\begin{equation}
    B_n \cdot (x_n -x_{n-1})=f(x_n)-f(x_{n-1}) \quad 
(\Rightarrow B_n\cdot s_{n-1} = y_{n-1})  \label{eqn:Bn_secante}
\end{equation}
By this condition alone, $B_n$ is obviously underdetermined.
To mitigate this issue,  $B_n$ is taken as a one-dimensional modification of $B_{n-1}$
satisfying \eqref{eqn:Bn_secante}. Concretely,
a sequence $(u_n)_{n \in \N} \in (K^m)^\N$ is introduced such that:
\begin{align}
    B_n &=   B_{n-1}+(y_{n-1}-B_{n-1} s_{n-1}) \cdot u_{n-1}\tran. \label{eqn:def_Bn} \\
     1 &= u_{n-1}\tran \cdot s_{n-1}  \label{eqn:unsn_eq_1}.
\end{align}
In Broyden's method over $\R$, $u_{n-1}$ is defined by:
\begin{equation}
    u_{n-1} = \frac{s_{n-1}}{ s_{n-1} \tran \cdot s_{n-1}  }. \label{eqn:un_over_R}
\end{equation}
The computation of the inverse of $B_n$
can then be done using the Sherman-Morrison formula (see \cite{SM49}):
\begin{equation}
  B_n^{-1} = B_{n-1}^{-1}+\frac{(s_{n-1}-B_{n-1}^{-1} y_{n-1})\cdot 
s_{n-1}\tran B_{n-1}^{-1}} { s_{n-1}\tran B_{n-1}^{-1}y_{n-1}} .  \label{eqn:comp_inv_Bn}
\end{equation}

This formula gives rise to the so-called ``good Broyden's method''. Using \cite{SM49} provides the following alternative formulae:
\begin{align}
    B_n &=   B_{n-1}+f_n \cdot u_{n-1}\tran. \\ 
B_n^{-1} &= B_{n-1}^{-1}-\frac{B_{n-1}^{-1} f_{n}\cdot 
u_{n-1}\tran B_{n-1}^{-1}} {u_{n-1}\tran B_{n-1}^{-1}y_{n-1}} . \label{eqn:comp_inv_Bn_with_un_variante_yn1}
\end{align}
\vspace{-1mm}
\vspace{-1mm}
\vspace{-1mm}
\vspace{-1mm}
\vspace{-1mm}
\subsection{Convergence}
We recall some notions on convergence of sequences
commonly used in the analysis of the behavior of Broyden's method.
\begin{defn}[\cite{OR} Chapter 9]
A sequence $(x_k)_{k \in \N} \in (K^m)^\N$
has Q-order of convergence $\mu\in \R_{> 1}$ to a limit $\xs \in K^m$,
if:
\[
\exists r \in \R_+, \quad \forall k \text{ large enough,} 
\quad \frac{\Vert x_{k+1}-\xs \Vert}{\Vert x_k-\xs \Vert^\mu } \leq r. 
\]
If we can take $\mu=1$ and $r<1$ in the previous inequality,
we say that $(x_k)_{k \in \N} $ has Q-linear convergence.
For $\mu=2$, we say it has Q-quadratic convergence.
The sequence is said to have Q-superlinear convergence if \vspace{-.2cm}
\[\lim_{k \rightarrow + \infty } \frac{\Vert x_{k+1}-\xs \Vert}{\Vert x_k-\xs \Vert } =0. \]
It is said to have R-order of convergence\footnote{R-convergence is a weaker notion, aimed at sequences not monotonically decreasing.} $\mu \in \R_{\geq 1}$
if
\[\limsup \Vert x_k-\xs \Vert^{ 1 /\mu^k}<1.\] \label{defn:Q_R_CV}
\end{defn} \vspace{-.5cm}
\begin{rmk}
For both Q and R, 
we write \textit{has convergence} $\mu$ to mean \textit{has convergence at least} $\mu.$
\end{rmk}
Broyden's  method  satisfies the following convergence results:
\begin{theo}
Over $\R^m,$ under usual regularity assumptions, Broyden's method defined by 
Eq.~\eqref{eqn:un_over_R} converges locally\footnote{By locally, we mean 
that for any $x_0$ and $B_0$ in small enough
balls around $\xs$ and $f'(\xs)$, the following convergence property
is satisfied.} Q-superlinearly \cite{BDM},
exactly in $2m$ steps for linear systems,
and with R-order at least $2^{\frac{1}{2m}}>1$ \cite{Gay79}. 
\end{theo}

Unfortunately, for general $K,$
Eq.~\eqref{eqn:un_over_R} is not a good fit.
Indeed, the quadratic form 
$x\mapsto x\tran x$ can be isotropic over $K^m$,
\textit{i.e.} there can be an $s_n \neq 0$ such that $s_n \tran \cdot s_n=0$.
This is the case, 
for example if $s_n = (X,X)$ in $\F_2 \llbracket X \rrbracket^2.$
Consequently, \eqref{eqn:un_over_R} has to be modified.
Trying to seek for another quadratic form 
that would not be isotropic is pointless, since for example there is none
over $\Qp^m$  for $m\ge5$~\cite{CoursSerre}.

\begin{rmk}
In the sequel, all the $B_i$'s will be invertible matrices.
Consequently, $s_{n+1}=0$ if and only if $f(x_{n})=0.$
We therefore adopt 
 the convention that if for some
$x_n,$ we have $\ f(x_n)=0,\ $
then the
sequences $(x_v)_{v\geq n}$ and 
$(B_v)_{v\geq n}$
will be  constant, and this case
does not require any further development. \label{rem:zero_and_constancy}
\end{rmk}

\section{Non-archimedean adaptation}
\label{sec:adaptation}
\subsection{Norms}

We use the following natural (non-normalized) norm
on $K$ defined from its valuation: for any $x \in K,$ 
$\Vert x \Vert = 2^{-\val (x)},$
except for $K=\Qp,$ where we take the more natural
$p^{-\val(x)}$ over $\Q_p$. 
Our norm\footnote{Over $\R$, it is of course
denoted by $\Vert \cdot \Vert_{\infty}$, but when based on  a non-archimedean
absolute value, this notation
is not used since it is implicitly unambiguous:
other norms such as the $\Vert \cdot \Vert_p$
are mostly useless.} on $K$ can naturally be extended to $K^m$: for any $x=(x_1,\dots, x_m) \in K^m,$
$\Vert x \Vert = \max_i \vert x_i \vert.$ We denote by $\val (x)$ the minimal valuation
among the $\val (x_i)$'s. It defines the norm of $x.$

\begin{lem}
Let $\biginterleave \cdot \biginterleave$ be the norm on $M_m(K)$
induced by $\Vert \cdot \Vert$. Let us abuse notations by
denoting with $\Vert \cdot \Vert$ the \textit{max}-norm
on the coefficients of the matrices of $M_m(K)$.
Then $\biginterleave \cdot \biginterleave=\Vert \cdot \Vert.$ \label{lem:Matrix_norms}
\end{lem}
\begin{proof}
Let $A \in M_n(K).$
If $x \in K^m $ is such that $\Vert x \Vert \leq 1,$
then by ultrametricity, it is clear that $\Vert Ax \Vert \leq \Vert A \Vert,$ hence $\biginterleave A \biginterleave \leq \Vert A \Vert.$
If $i \in \N$ is such that $\Vert A \Vert$ is obtained with a coefficient
on the column of index $i$, then $\Vert A e_i \Vert = \Vert A \Vert$,
whence the equality.
\end{proof}

Consequently, the \textit{max}-norm on the coefficients
of a matrix is a matrix norm.
For rank-one matrices, the computation of 
the norm can be made easy using the following corollary of Lemma \ref{lem:Matrix_norms}.

\begin{cor}
Let $a,b  \in K^m$ be two vectors.
Then
\begin{equation}
    \Vert  a \hspace{.3mm} \tran \cdot b \Vert = \Vert a \Vert \cdot \Vert b \Vert.
\end{equation} \label{cor:rankone_norm}
\end{cor}

\vspace{-1mm}
\vspace{-1mm}
\vspace{-1mm}
\vspace{-1mm}
\vspace{-1mm}
\vspace{-1mm}
\vspace{-1mm}
\vspace{-1mm}
\subsection{Constraints and optimality}

For the sequence $(x_n)_{n \in \N}$ to be well defined,
the sequence $(u_n)_{n \in \N}$ must satisfy
Eqs~\eqref{eqn:def_Bn}-\eqref{eqn:unsn_eq_1} and also:
\begin{equation}
    s_n\tran B_{n}^{-1}y_n \neq 0, 
\end{equation}
to ensure Eq.~\eqref{eqn:comp_inv_Bn} makes sense.
Many different $u_n$'s can satisfy those conditions.
Over $\R,$ Broyden's choice of $u_n$ defined by \eqref{eqn:un_over_R}
can be characterized by minimizing the Frobenius norm of
$B_{n+1} - B_{n}.$
We can proceed similarly over $K.$

\begin{lem}
If $B_{n+1}$ satisfies \eqref{eqn:Bn_secante}, then:
\begin{equation}
\Vert B_{n + 1} -B_{n} \Vert \geq \frac{\Vert y_n-B_n s_n \Vert }{\Vert s_n \Vert}. \label{eqn:norm_bn_minus_bn1}
\vspace{-1mm}
\vspace{-1mm}
\vspace{-1mm}
\end{equation} \label{lem:norm_bn_minus_bn1}
\vspace{-1mm}
\end{lem}
\vspace{-1mm}
\begin{proof}
It is clear as in this case, $(B_{n + 1} -B_{n})s_n=y_n-B_n s_n.$
\end{proof}
This inequality  can become an equality with a suitable choice
of $u_n$ as shown in the following lemma.
\begin{lem}\label{lem:1}
Let $l$ be such that $\val (s_{n,l})=\val (s_n). $
Then \[u_n = s_{n,l}^{-1} e_l\] satisfies \eqref{eqn:unsn_eq_1}
and reaches the bound in \eqref{eqn:norm_bn_minus_bn1}.
\end{lem}
Nevertheless, this is not enough to have $B_n$
invertible in general, as we can see from the 
Sherman-Morrison formula~\eqref{eqn:comp_inv_Bn_with_un_variante_yn1}:
\begin{lem}
$B_n$ defined by Eq.\eqref{eqn:def_Bn} is invertible
if and only if 
\begin{equation}
    u_{n-1} \tran B_{n-1}^{-1} y_{n-1} \neq 0. \label{eqn:cond_sherman_morrison}
\end{equation}
\end{lem}
The next lemma shows how to choose $l$, up to the condition
$ (B_{n-1}^{-1} y_{n-1})_l \neq 0$,
which actually never occurs after Corollary~\ref{cor:one_case}.
\begin{lem}
Let $l$ be the {\em smallest} index such that $\val (s_{n,l})=\val (s_n). $
If $ \left( B_{n-1}^{-1} y_{n-1} \right)_l \neq 0,$ then
\vspace{-8pt}
\begin{equation}
    u_n = s_{n,l}^{-1} \, e_l \label{eqn:def_simple_un}
\end{equation}
\vspace{-4pt}
 satisfies Eq.~\eqref{eqn:unsn_eq_1},
reaches the bound in Eq.~\eqref{eqn:norm_bn_minus_bn1} and satisfies
Eq.\eqref{eqn:cond_sherman_morrison}.
\end{lem}

\section{Local Convergence}
\label{sec:convergence}
\subsection{Local Linear convergence}

Let $E$ and $F$ be two finite-dimensional normed vector spaces over
$K$ We denote by $L(E,F)$ the space of $K$-linear mappings 
from $E$ to $F$.

\begin{defn}\label{defn:diff}
Let $U$ be an open subset of $E$.
A function $f : U \to F$ is \emph{strictly differentiable} at $x \in U$ if there exists 
an $f'(x) \in L(E,F)$ satisfying the
following property: for all $\varepsilon > 0$, there 
exists a neighborhood $U_{x,\varepsilon} \subset U$ of $x$, on which
for any $y,\ z\in U_{x,\epsilon}$:
\begin{equation}
\Vert f(z) - f(y) - f'(x) \cdot (z{-}y)\Vert_F \leq  \varepsilon \cdot 
\Vert z{-}y\Vert_E.
\end{equation}
\end{defn}
Note that both $z$ and $y$ can vary.
This property is natural in the ultrametric 
context (see 3.1.3 of \cite{caruso2017}), as
the counterpart of Fr\'echet differentiability over $\mathbb{R}$
does not provide meaningful local information.
Polynomials and converging power series satisfy
strict differentiability everywhere they are defined.

We can then adapt Theorem 3.2 of \cite{BDM} in our ultrametric setting.

\begin{theo}
Let $f : K^m \to K^m$ and $\xs \in U$ be such that $f$ is strictly differentiable at $\xs,$
$f'(\xs)$ is invertible and $f(\xs)=0.$
Then any quasi-Newton method whose choice of $u_n$
yields
for all $n,$ $\Vert u_n \Vert = \Vert s_n \Vert^{-1}$ (which includes Broyden's choice of Eq. \eqref{eqn:def_simple_un}), 
is locally $Q$-linearly converging to $\xs$ with ratio $r$ for
any $r \in (0,1).$ \label{thm:conv_lineaire}
\end{theo}
\begin{proof}
Let $r \in (0,1).$
Let the constants $\gamma, \delta, $ and $\lambda$ be satisfying:
\begin{equation}\label{eqn:trois}
\hspace{-12pt}
\gamma \geq \Vert f'(\xs)^{-1} \Vert,
\quad 
    0<\delta \leq  \frac{r}{\gamma (1+r)(3-r)},
    \quad
    0<\lambda \leq  \delta (1-r).
\end{equation}
Let $\eta>0$ be given by the strict differentiability at $\xs$
and such that on the ball $B(\xs,\eta),$
\[\Vert f(z) - f(y) - f'(\xs) \cdot (z{-}y)\Vert \leq  \lambda \cdot 
\Vert z{-}y\Vert. \]
We restrict further $\eta$ so as to have: $\eta \leq \delta (1-r).$
Let us assume that 
$$
    \Vert B_0 -f'(\xs) \Vert \leq \delta, 
\qquad \Vert x_0 - \xs \Vert < \eta.
$$

We have from the condition on $\delta$ that 
$\delta \, \gamma (1+r)(3-r) \leq r.$
Since $3-r>2,$ 
then $ 2 \delta\, \gamma (1+r) \leq r.$
Consequently,
\[
\frac{1}{1-2 \delta \gamma} \leq 1+r, 
\]
the denominator being non zero
because $\delta < (2 \gamma)^{-1}.$

Since $\Vert f'(\xs)^{-1} \Vert \leq \gamma$
and $\Vert B_0-f'(\xs) \Vert  < 2 \delta,$
the Banach Perturbation Lemma (\cite{OR} page 45) in the
Banach algebra $M_m(K)$ implies that $B_0$ is invertible and:
\[\Vert B_0^{-1} \Vert \leq \frac{\gamma}{1-2 \gamma \delta} \leq (1+r) \gamma. \] 

We can now estimate what happens to $x_1 = x_0 - B_0^{-1} f(x_0).$

\begin{align}
\Vert x_1-\xs \Vert &=\Vert x_0-\xs -B_0^{-1} f(x_0) \Vert, \label{eqn:first}\\
\nonumber &=\Vert-B_0^{-1} \left(f(x_0)-f(\xs)-f'(\xs)\cdot(x_0-\xs) \right) \\
\nonumber & \: \:
-B_0^{-1} \left( f'(\xs) (x_0 - \xs)-B_0(x_0-\xs) \right) \Vert ,\\
\nonumber &=\Vert-B_0^{-1} \left(f(x_0)-f(\xs)-f'(\xs)\cdot(x_0-\xs) \right)\\
\nonumber & \: \:
-B_0^{-1} \left( (f'(\xs)-B_0) (x_0 - \xs)\right) \Vert ,\\
\nonumber &\leq \Vert B_0^{-1} \Vert \left( \lambda \Vert x_0 - \xs \Vert +2\delta \Vert x_0 - \xs \Vert \right) , \\
\nonumber &\leq \Vert B_0^{-1} \Vert (\lambda + 2 \delta )  \Vert x_0 - \xs \Vert, \\
\nonumber& \leq \gamma (1+r) ( \delta (1-r)+2 \delta) \Vert x_0 - \xs \Vert, \\
\nonumber & \leq \gamma (1+r) \delta (3-r) \Vert x_0 - \xs \Vert \qquad  \text{by Eq.~\eqref{eqn:trois} (middle)} \\
& \leq r \Vert x_0 - \xs \Vert. \label{eqn:last}
\end{align}

Consequently, $\Vert x_1-\xs \Vert \leq  r \Vert x_0 - \xs \Vert$
and $\Vert x_1-\xs \Vert \leq r \eta <\eta,$ \textit{i.e.} $x_1 \in B(\xs, \eta).$

Eq.~\eqref{eqn:def_Bn} defines $B_1$ by $B_1=B_0-(y_1-B_0 s_1) \cdot u_1\tran$ for some
$u_1$  verifying $\Vert u_1 \Vert = \Vert s_1 \Vert^{-1}$ (see Eqs.~\eqref{eqn:unsn_eq_1}, Corollary~\ref{cor:rankone_norm}). Then: 
\vspace{-1mm}
\[\Vert B_1 -B_0 \Vert= \Vert f(x_{1})-f(x_0)-B_0 (x_1-x_0) \Vert \cdot \Vert x_1-x_0 \Vert^{-1}.\]
\vspace{-1mm}
Therefore,
\begin{align}
\Vert B_1 - f'(\xs) \Vert & \leq \max \left( \phantom{ {}^{-1}} \Vert B_0-f'(\xs) \Vert \right. ,\label{eqn:B1}\\
\nonumber & \: \: \: \: \: \left. \Vert f(x_{1})-f(x_0)-B_0 (x_1-x_0) \Vert \Vert x_1-x_0 \Vert^{-1}  \right), \\
\nonumber & \leq \max \left( \phantom{ {}^{-1}} \Vert B_0-f'(\xs) \Vert \right. ,\\
\nonumber & \: \: \: \: \: \Vert \left( B_0-f'(\xs) \right) (x_1-x_0) \Vert \Vert x_1-x_0 \Vert^{-1}, \\
\nonumber & \: \: \: \: \: \left. \Vert f(x_{1})-f(x_0)-f'(\xs) (x_1-x_0) \Vert \Vert x_1-x_0 \Vert^{-1}  \right), \\
\nonumber & \leq \max (\delta, \lambda ) \leq \delta.
\end{align}

We can then carry on and prove by induction
that for all $k$,
\begin{equation}\label{eqn:r^k}
\hspace{-6pt}\text{(i)}\ \ \Vert x_{k} - \xs \Vert \leq r^k \Vert x_0-\xs \Vert,
\quad
\text{and}
\quad \text{(ii)}
\ B_k \in B(f'(\xs), \delta).
\end{equation}

Heredity for Inequality~\eqref{eqn:r^k}-(i) comes from: a same use of the Banach Perturbation Lemma on $B_k$ so that $B_k$ is invertible;
that $\Vert B_k^{-1} \Vert \leq (1+r)\gamma$ and  by repeating the computations ~\eqref{eqn:first} to \eqref{eqn:last}:
\vspace{-.2cm}
\begin{align*}
\Vert x_{k+1}-\xs \Vert &\leq \Vert B_k \Vert^{-1} (\lambda+2 \delta) \Vert x_k-\xs \Vert, \\
& \leq (1+r)\gamma \delta (3-r) \Vert x_k-\xs \Vert, \\
& \leq r \Vert x_k - \xs \Vert.\\
\end{align*} 

\vspace{-.15cm}
\vspace{-.15cm}
\vspace{-.15cm}
\vspace{-.15cm}

We can deal with~\eqref{eqn:r^k}-(ii)  using  a similar computation as~\eqref{eqn:B1}:
\vspace{-1mm}
\begin{align}
\label{eqn:Bk+1} &\Vert B_{k+1} -f'(\xs) \Vert 
\leq \  \max \big( \Vert B_k-f'(\xs) \Vert , \\
\nonumber & \quad \qquad \Vert f(x_{k+1})-f(x_k)-B_k (x_{k+1}-x_k) \Vert \Vert x_{k+1}-x_k \Vert^{-1}  \big) \\
\nonumber & \qquad \leq \max \big( \Vert B_k-f'(\xs) \Vert , \\
\nonumber &  \quad \qquad \Vert f(x_{k+1})-f(x_k)-f'(\xs) (x_{k+1}-x_k) \Vert \Vert x_{k+1}-x_k \Vert^{-1}  \big), \\
 & \qquad \leq \max (\delta, \lambda )  \leq \delta.   \qedhere 
\end{align}
\end{proof}

\begin{cor}\label{cor:one_case}
Locally, one can 
take definition \eqref{eqn:def_simple_un} to
define all the $u_n$'s and all the $B_n$'s
will still be invertible. 
\end{cor}
\begin{proof}
With the assumptions of the proof of Theorem \ref{thm:conv_lineaire}, for $u_n$  defined by \eqref{eqn:def_simple_un}, $\Vert u_{n-1} \Vert = \Vert s_{n-1} \Vert^{-1}$ and \eqref{eqn:unsn_eq_1} are satisfied,
and by the Banach Perturbation Lemma, $B_n$ defined by  \eqref{eqn:def_Bn}
is invertible.
\end{proof}

\begin{rmk}
The fact that Broyden's method has locally
Q-linear\end{rmk}

\noindent convergence with ratio $r$
for any $r$ is not enough to
prove that ithas Q-superlinear convergence.
Indeed, as $x_k$ is going closer to
$\xs,$ there is no reason for
$B_k$ to get closer to $f'(\xs).$
Consequently, we cannot
expect from the previous result that
$x_k$ and $B_k$
enter loci of smaller ratio of convergence
as $k$ goes to infinity.
In fact, in general, $B_k$ 
does not converge to $f'(\xs).$

Finally, the next lemma, consequence
of the previous theorem, 
will be useful in the next subsection 
to obtain the R-superlinear convergence.
\begin{lem}
Using the same notations as in the proof of 
Theorem \ref{thm:conv_lineaire},
if $r \leq \left(\frac{\gamma \Vert f'(\xs ) \Vert}{2}\right)^{-1},$
and $\Vert B_0-f'(\xs) \Vert < \delta$ 
and $\Vert x_0 - \xs \Vert < \eta,$
then for all $n \in \N,$
\vspace{-2mm}
\[\Vert f_{n+1} \Vert \leq \Vert f_n \Vert. \] \label{lem:decr_norm_fn}
\vspace{-2mm}
\end{lem} 
\vspace{-2mm}
\vspace{-2mm}
\begin{proof}
Let $n \in \N.$
We have $\Vert s_n\Vert \le r \Vert s_{n-1}\Vert$. Indeed,
from $\Vert x_{n+1}- x_n\Vert \le \max (\Vert x_{n+1} - x^\star\Vert, \ \Vert x^\star - x_n\Vert)$,
and $\Vert x_{n+1}-x_n \Vert < \Vert x_n -x^\star \Vert$, we see that
 $\Vert s_n\Vert = \Vert x^\star - x_n\Vert \le r \Vert x^\star - x_{n-1}\Vert = r \Vert s_{n-1}\Vert$.

Then using \eqref{eqn:QNewton} and 
the Q-linear convergence with ratio $r,$
we get that 
$\Vert f_{n+1} \Vert \leq r \Vert B_{n+1} \Vert \Vert B_n^{-1} \Vert \Vert f_n \Vert.$
Using~\eqref{eqn:Bk+1}, the definition of $\delta,$ $\gamma$ in~\eqref{eqn:trois},
and the fact that $0 <r <1,$
we get that $\Vert B_{n+1} \Vert \Vert B_n^{-1} \Vert \leq 2 \gamma \Vert f'(\xs ) \Vert ,$ which concludes the proof.
\end{proof}

\subsection{Local R-superlinear convergence}

We first remark that the $2n$-step convergence
in the linear case proved by Gay in \cite{Gay79}
is still valid. Indeed, it is only a matter of 
linear algebra.

\begin{theo}[Theorem 2.2 in \cite{Gay79}]
If $f$ is defined by $f(x)=Ax-b$ for some $A \in GL_m (K)$,
then any quasi-Newton method converges in at most $2m$ steps (\textit{i.e.} $f(x_{2m})=0$).
\end{theo}

With this and under a stronger differentiability assumption on $f$,
we can obtain R-superlinearity,
similarly to Theorem 3.1 of \cite{Gay79}. 
The proof also follows the main steps thereof.

\begin{theo}
Let us assume that on a neighborhood $U$ of $\xs,$
there is a $c_0 \in \R_{>0}$ such that
$f$ satisfies\footnote{This condition is satisfied by polynomials or converging power series.} 
\begin{equation}
    \forall x,y \in U, \Vert f(x)-f(y) - f'(\xs) \cdot (x-y) \Vert \leq c_0 \Vert x-y \Vert^2.  \label{DL2}
\end{equation}

Then there are $\eta, \: \delta$ and $\Gamma$ in $\R_{>0}$ such that if $x_0 \in B(\xs, \eta)$ 
and $B_0 \in B(f'(\xs), \delta),$ then for any $w \in \Z_{\geq 0},$
\[\Vert x_{w+2m}-\xs \Vert \leq \Gamma \Vert x_w - \xs \Vert^2. \] \label{thm:R_order}
\end{theo}
\vspace{-0.7cm}
\begin{proof} 
{\em Step~1: Preliminaries.}
Condition \eqref{DL2} is stronger than strict differentiability as stated in
Theorem \ref{thm:conv_lineaire}. 
From its proof and Lemma \ref{lem:decr_norm_fn}, 
let $r \in (0,1)$ and
$\gamma \geq \Vert f'(\xs)^{-1} \Vert,$
as well as $\eta$ and $\delta$ such that:
$r\leq \left(\frac{\gamma \Vert f'(\xs ) \Vert}{2}\right)^{-1},$
and
if $x_0 \in B(\xs, \eta)$ 
and $B_0 \in B(f'(\xs), \delta),$
the sequences  $(x_n)_{n \in \N}$ and $(B_n)_{n \in \N}$ 
defined by  Broyden's method
(using \eqref{eqn:def_simple_un})
are well defined and moreover the four following inequalities
are satisfied:
for any $k \in \N,$
\begin{displaymath}
\begin{array}{ll}
   \Vert B_k -f'(\xs) \Vert \leq \delta,   & \Vert x_{k+1} - \xs \Vert \leq r \Vert x_k-\xs \Vert,  \\
  \Vert B_k^{-1} \Vert \leq (1+r) \gamma ,   & \Vert f(x_{k+1}) \Vert \leq \Vert f(x_k) \Vert.
\end{array}
\end{displaymath}

Let   $x_0 \in B(\xs, \eta),$ 
 $B_0 \in B(f'(\xs), \delta),$
 and $(x_n)_{n \in \N}$ and $(B_n)_{n \in \N}$ 
be defined by Broyden's method.
Let $w \in \N$ and $h = \Vert x_w - \xs \Vert.$
We must show that there is a $\Gamma,$
independent of $w$ such that $\Vert x_{w+2m} - \xs  \Vert \leq \Gamma h^2.$

{\em Step~2: reference to a linear map.}
Let the linear affine map $\hat{f}(x) =f'(\xs) \left(x-\xs \right),$
and $\hat{x_0}= x_w$ and $\hat{B}_0=B_w.$
Broyden's method
(using first \eqref{eqn:def_simple_un})
applied to those data
produces the  sequences
$(\hat{x}_n)_{n \in \N}$ and $(\hat{B}_n)_{n \in \N},$ 
which are constant for $n \geq 2m,$
as a result of Theorem \ref{thm:conv_lineaire}.
We define similarly  $\hat{s}_n=\hat{x}_{n+1}- \hat{x}_n$.
We have again for all $k \in \N$ the four inequalities:
\begin{displaymath}
\begin{array}{ll}
   \Vert \hat{B}_k -f'(\xs) \Vert \leq \delta,   & \Vert \hat{x}_{k+1} - \xs \Vert \leq r \Vert \hat{x}_k-\xs \Vert,  \\
  \Vert \hat{B}_k^{-1} \Vert \leq (1+r) \gamma   & \Vert \hat{f}(x_{k+1}) \Vert \leq \Vert \hat{f}(x_k) \Vert.
\end{array}
\end{displaymath}
The key to the proof is that $\hat{x}_{2m}=\xs$ and $\hat{x}_k$ and $x_{w+k}$ are not
too much far apart.

{\em Step~3: Statement of the induction.}
More concretely, we prove by induction on $j$ that there exist $\gamma_{1,j}$ and $\gamma_{2,j},$
independent of $w$, such that for $0 \leq j \leq 2m,$ we have the two
inequalities:
\begin{equation}
    \Vert B_{w+j} - \hat{B}_j \Vert \cdot \Vert f_{w+j} \Vert \leq \gamma_{1,j} h^2, \tag{$E_{1,j}$} \label{eqn:E_1j}
\end{equation}
\begin{equation}
    \Vert x_{w+j} - \hat{x}_j \Vert  \leq \gamma_{2,j} h^2. \tag{$E_{2,j}$} \label{eqn:E_2j}
\end{equation}

{\em Step~4: Base case.}
Since $B_w = \hat{B}_0$ and $x_w = \hat{x}_0,$ $(E_{1,0})$ and $(E_{2,0})$ are clear, 
with $\gamma_{1,0}=\gamma_{2,0}=0.$
Now, let us assume that $(E_{1,k})$ and $(E_{2,k})$ are true for a given $k$
such that $0\leq k <  2m.$
\smallskip

{\em Step~5:} We first prove $(E_{2,k+1}).$ One part of the inequality~\eqref{eqn:hard}
is obtained thanks to: $B_{w+k}^{-1} - \hat{B}_k^{-1} = B_{w+k}^{-1}( \hat{B}_k - B_{w+k}) \hat{B}_k^{-1}$.
\begin{align}
\nonumber    \Vert s_{w+k}-\hat{s}_k \Vert   & = 
    \Vert B_{w+k}^{-1} f_{w+k} - \hat{B}_k^{-1}\hat{f}(\hat{x}_k) \Vert \\
        \label{eqn:hard}                            & \hspace{-16pt}\leq \max \left( \Vert B_{w+k}^{-1} \Vert \cdot 
    \Vert \hat{B}_k^{-1} \Vert \cdot \Vert B_{w+k} - \hat{B}_k \Vert \cdot \Vert f_{w+k} \Vert ,\right. \\ 
         \nonumber                           & \left. \Vert \hat{B}_k^{-1} \Vert \cdot \Vert f_{w+k}-\hat{f}(\hat{x}_k) \Vert \right) \\
                  \nonumber                  &\hspace{-16pt}\leq \Vert \hat{B}_k^{-1} \Vert \max  
    \left( \Vert B_{w+k}^{-1} \Vert \cdot \Vert B_{w+k}-\hat{B}_k \Vert\cdot  \Vert f_{w+k} \Vert \right. , \\ 
\label{eqn:last3}                                &   \left. \Vert f_{w+k}-\hat{f}(x_{w+k}) \Vert \ , \  \Vert \hat{f}(x_{w+k})-\hat{f}(\hat{x}_k) \Vert \right)
\end{align}
The first term on the r.h.s. of~\eqref{eqn:last3} is upper-bounded by $(1+r)^2 \gamma^2 \gamma_{1,k} h^2$ using $(E_{1,k})$
and $ \Vert B_{w+k}^{-1}  \Vert \le (1+r)\gamma$.

For the second term of~\eqref{eqn:last3}, using~\eqref{DL2}:
\[\Vert f_{w+k}-f(\xs) -f'(\xs) \cdot (x_{w+k}-\xs) \Vert \leq  c_0 \Vert x_{w+k}-\xs \Vert^2\]
and $\Vert x_{w+k}-\xs \Vert \leq \Vert x_w-\xs \Vert = h,$
it is upper-bounded by $c_0 h^2.$
Finally, the last term is equal to $f'(\xs)(x_{w+k} - \hat{x}_k)$
whose norm is upper-bounded  by 
$ \Vert f'(\xs) \Vert \gamma_{2,k} h^2$
thanks to $(E_{2,k})$.
This is enough to define $\gamma_{3,k}$
such that $\Vert s_{w+k}-\hat{s}_k \Vert \leq \gamma_{3,k} h^2\quad \scriptstyle{(\ddagger)}$.
Consequently, with $\gamma_{2,k+1}= \max( \gamma_{3,k} ,  \gamma_{2,k}),$
we do have $ \Vert x_{w+k+1} - \hat{x}_{k+1} \Vert  \leq \gamma_{2,k+1} h^2,$
and $(E_{2,k+1})$ is satisfied.

{\em Step~6.0:} We now prove $(E_{1,k+1}).$
We first deal with some preliminary cases.
If $s_{w+k}=0,$ (that is $x_{w+k+1} = x_{w+k}$) 
then the property~\eqref{eqn:Bn_secante}
$s_{w+k} = - B_{w+k}^{-1} f_{w+k}$ implies that $f_{w+k}=0$,
and the property $B_{w+k+1} s_{w+k} = y_{w + k}$ implies that
$f_{w+k} = f_{w+k+1}=0$.
Thus $(E_{1,k+1})$ is satisfied with $\gamma_{1,k+1} = 0.$
If $\hat{s}_k=0,$ then similarly $\hat{f}(\hat{x}_{w+k})=\hat{f}(\hat{x}_{w+k+1})=0$.
Therefore, as we have seen before,
\begin{align*}
 \Vert f_{w+k+1} \Vert &= \Vert f_{w+k+1}-\hat{f}(x_{w+k+1})+\hat{f}(x_{w+k+1})-\hat{f}(\hat{x}_{k+1}) \Vert, \\
 &\leq \max \left(c_0 , \Vert f'(\xs) \Vert \gamma_{2,k+1}  \right)h^2.
\end{align*}
Then, using that $\Vert B_{w+k+1} - \hat{B}_{k+1} \Vert \leq \max ( \Vert B_{w+k+1}- f'(\xs)  \Vert ,
\Vert \hat{B}_{k+1} - f'(\xs) \Vert ) \le \delta,$ $(E_{1,k+1})$ is satisfied
with:\\
\centerline{$\gamma_{1,k+1} = \delta h^2 \max\left(c_0 , \Vert f'(\xs) \Vert \gamma_{2,k+1}  \right)$.}

{\em Step~6.1 :} We can now assume that both $s_k$ and $\hat{s}_k$ are non zero.
To prove that there is a $\gamma_{1,k+1}$ (independent of $w$)
such that $(E_{1,k+1})$ holds,
then in view of the fact that $\Vert f_{w+k+1} \Vert \leq \Vert f_{w+k}\Vert$ (Lemma~\ref{lem:decr_norm_fn})
of $(E_{1,k})$ and of the definition (Eq.~\eqref{eqn:def_Bn}) of $B_{k+1}$ and $\hat{B}_{k+1}$,
it is enough to prove that there is some  $\gamma_{4,k+1}$ (independent of $w$)
such that:
\begin{multline}\label{eqn:gamma4}
  \Vert \left( y_{w+k}-B_{w+k}s_{w+k} \right) u_{w+k} \tran - \\
  \left( \hat{y}_{k} -
 \hat{B}_{k}\hat{s}_{k} \right) \hat{u}_{k}\tran \Vert \cdot \Vert f_{w+k+1} \Vert  \leq  \gamma_{4,k+1} h^2.
\end{multline}
Using that $\Vert f_{w+k+1} \Vert \leq \Vert f_{w+k} \Vert$ (by Lemma~\ref{lem:decr_norm_fn}),
we obtain:
\begin{align}
\nonumber   \Vert f_{w+k+1}\Vert & 
 \cdot     \Vert \left( y_{w+k}-B_{w+k}s_{w+k} \right) u_{w+k} \tran -
   \left( \hat{y}_{k}-\hat{B}_{k}\hat{s}_{k} \right) \hat{u}_{k}\tran \Vert  \\
 \nonumber    \leq & \Vert f_{w+k} \Vert \max \left( \Vert y_{w+k} - f'(\xs) s_{w+k} \Vert \cdot \Vert u_{w+k} \tran \Vert \right. ,  \\
 \nonumber    & \left.  \Vert (f'(\xs)-B_{w+k}) s_{w+k} u_{w+k} \tran - (f'(\xs)-\hat{B}_k)\hat{s}_k \hat{u}_k\tran \Vert \right)\\
 \label{eqn:i}  \qquad  \leq & \Vert f_{w+k} \Vert \max \left( \Vert y_{w+k} - f'(\xs) s_{w+k} \Vert \cdot \Vert u_{w+k} 
     \tran \Vert \right. , \\
  \label{eqn:ii}  \qquad    &  \Vert (f'(\xs)-\hat{B}_k) 
    (s_{w+k} u_{w+k} \tran - \hat{s}_k \hat{u}_k\tran ) \Vert ,  \\ 
  \label{eqn:iii} \qquad    &  \left. \Vert (B_{w+k}-\hat{B}_k)
    s_{w+k} u_{w+k} \tran \Vert \right) . 
\end{align}

{\em Step~6.2:} From $f_{w+k} = - B_{w+k}s_{w+k},$ we have $\Vert f_{w+k} \Vert \le \Vert s_{w+k} \Vert \cdot 
\max (  \Vert  B_{w+k} -f'(\xs) \Vert ,  
\Vert f'(\xs) \Vert) \le \Vert s_{w+k} \Vert \cdot \max(\delta, \Vert f'(\xs)\Vert)
\ \scriptstyle{(\bullet)}$.
\mbox{Otoh} by \eqref{DL2}, $\Vert y_{w+k} - f'(\xs) s_{w+k} \Vert \le  c_0 \Vert s_{w+k}\Vert^2 $.
It follows that the first term \eqref{eqn:i} can be upper-bounded in the following way:
$$
\hspace{-8pt}\eqref{eqn:i} \le c_0 \Vert s_{w+k}\Vert^3 \Vert u_{w+k}\tran \Vert \max( \delta, \Vert f'(\xs) \Vert )
\le c_0 h^2  
\max( \delta, \Vert f'(\xs) \Vert ),
$$
the rightmost inequality  being obtained from $\Vert u_{w+k}\tran \Vert = \Vert s_{w+k}\Vert^{-1}$
and $\Vert s_{w+k} \Vert \le \max (\Vert x_{w+k+1} -\xs \Vert, \Vert x_{w+k} -\xs \Vert) =
\Vert x_{w+k} -\xs \Vert \le \Vert x_{w} -\xs \Vert =h$.
 

 {\em Step~6.3:} The third one \eqref{eqn:iii} can be upper-bounded using $(E_{1,k})$:
\[\eqref{eqn:iii}  \le \Vert f_{w+k} \Vert \Vert (B_{w+k}-\hat{B}_k)s_{w+k} u_{w+k} \tran \Vert \leq \gamma_{1,k} h^2. \]

 {\em Step~6.4:} For the second one \eqref{eqn:ii}, observe that:
\begin{equation}\label{eqn:(ii)}
s_{w+k} u_{w+k}\tran - \hat{s}_k \hat{u}_k\tran= (s_{w+k}-\hat{s}_k)u_{w+k}\tran - \hat{s}_k (u_{w+k}\tran -\hat{u}_k \tran).
\end{equation}
The first term is easy to manage using the previous inequality $\scriptstyle{(\bullet)}$ on $\Vert f_{w+k}\Vert$,
the inequality $\scriptstyle{(\ddagger)}$ on $\Vert s_{w+k} - \hat{s}_k \Vert$ and $\Vert s_{w+k} \Vert \Vert u_{w+k}\tran \Vert =1$:
\begin{equation}\label{eqn:(ii)1}
\Vert f_{w+k} \Vert 
\cdot \Vert  (s_{w+k}-\hat{s}_k)u_{w+k}\tran \Vert \leq \max ( \delta  , \Vert f'(\xs) \Vert  )  \gamma_{3,k} h^2.
\end{equation}
The second one of Eq.~\eqref{eqn:(ii)} is a little bit trickier. Define as in \eqref{eqn:def_simple_un},
$u_{w+k}=s_{w+k,l}^{-1}e_l$ and $\hat{u}_{k}=\hat{s}_{k,\hat{l}}^{-1}e_{\hat{l}}$
for some given $l$ and $\hat{l}.$ 

If $l = \hat{l},$ we have: (the last inequality below follows from $\scriptstyle{(\ddagger)}$).
\begin{align*}
   \Vert u_{w+k}-\hat{u}_k \Vert &= \vert s_{w+k,l}^{-1}-\hat{s}_{k,l}^{-1} \vert 
    = \frac{\vert s_{w+k,l}-\hat{s}_{k,l} \vert}{\vert s_{w+k,l}\vert \cdot \vert \hat{s}_{k,l} \vert } 
     = \frac{\vert s_{w+k,l}-\hat{s}_{k,l} \vert}{\Vert s_{w+k}\Vert \cdot \Vert \hat{s}_{k} \Vert } \\
     & \leq \frac{\Vert s_{w+k}-\hat{s}_{k} \Vert}{\Vert s_{w+k}\Vert \cdot \Vert \hat{s}_{k} \Vert } 
     \leq \frac{\gamma_{3,k}h^2}{\Vert s_{w+k}\Vert \cdot \Vert \hat{s}_{k} \Vert }.
\end{align*}
From this and from  $\Vert f_{w+k}\Vert = \Vert B_{w+k} \Vert \cdot \Vert s_{w+k}\Vert$ we get:
\begin{equation}\label{eqn:(ii)2}
\Vert f_{w+k} \Vert  \cdot \Vert u_{w+k}-\hat{u}_k \Vert \cdot \Vert \hat{s}_{k} \Vert \leq \gamma_{3,k}
\max \left(\delta , \Vert f'(\xs) \Vert \right) h^2.
\end{equation}
If $l \neq \hat{l},$ then either $\Vert s_{w+k}-\hat{s}_k \Vert = \Vert s_{w+k} \Vert,$
if $\Vert \hat{s}_k \Vert \leq \Vert s_{w+k} \Vert,$
or $\Vert s_{w+k}-\hat{s}_k \Vert = \Vert \hat{s}_k  \Vert,$
if $\Vert s_{w+k} \Vert \leq \Vert \hat{s}_k  \Vert.$
In the first case, we have
\[\Vert u_{w+k}-\hat{u}_k \Vert = \Vert \hat{s}_k \Vert^{-1}, \]
and then, the second term of~\eqref{eqn:(ii)} multiplied by $\Vert f_{w+k}\Vert$ verifies: 
\begin{align}
 \nonumber   \Vert f_{w+k} \Vert \cdot  \Vert u_{w+k}-\hat{u}_k \Vert \cdot \Vert \hat{s}_{k} \Vert &\leq  
\max \left(\delta ,  \Vert f'(\xs) \Vert \right) \Vert s_{w+k} \Vert \\
    &\leq \max \left(\delta , \Vert f'(\xs) \Vert \right) \gamma_{3,k} h^2. \label{eqn:(ii)3}
\end{align}
The second case follows with the same computation.
Eqs~\eqref{eqn:(ii)3}~\eqref{eqn:(ii)2}~\eqref{eqn:(ii)1} prove together the
bound on the expression~\eqref{eqn:ii} in ~\eqref{eqn:(ii)}. In turn with the bounds on 
the terms \eqref{eqn:i} and \eqref{eqn:iii},
prove~\eqref{eqn:gamma4}.
This concludes the proof of 
$(E_{1,k+1}),$ and finally the induction.

{\em Step~7:} Consequently, $\Vert x_{w+2m} - \hat{x}_{2m} \Vert \leq \gamma_{2,2m} h^2.$
Thanks to Theorem \ref{thm:conv_lineaire}, $\hat{x}_{2m} = \xs,$
and thus, we have proved that for any $w,$
\[\Vert x_{w+2m} - \xs \Vert \leq \gamma_{2,2m} \Vert x_w - \xs \Vert^2.  \qedhere \]
\vspace{-.3cm}
\end{proof}
Theorem~\ref{thm:R_order} has for immediate consequence:
\begin{theo}\label{th:bigone}
Broyden's method has locally R-order of convergence $2^{\frac{1}{2m}}.$
\end{theo}
\begin{proof}
Let us take $x_0$ and $B_0$ as in the
proof of the previous theorem, and same constants and notations.
For any $w,$
$\Vert x_{w+2m} - \xs \Vert \leq \Gamma \Vert x_w - \xs \Vert^2. $

Consequently, for $0 \leq k <2m,$  $l \in \N,$ and $\mu = 2^{1/2m},$
\begin{align*}
    \Vert x_{2lm+k} - \xs \Vert^{\mu^{-2lm-k}} & \leq \Vert x_k -\xs \Vert^{2^l \mu^{-2lm-k}} \Gamma^{(2^l-1)\mu^{-2lm-k}}  \\
    &\leq \Vert x_k -\xs \Vert^{2^l 2^{-l-\frac{k}{2m}}} \Gamma^{(2^l-1)2^{-l-\frac{k}{2m}}} \\
    &\leq \Vert x_k -\xs \Vert^{2^{-\frac{k}{2m}}}\Gamma^{(1-2^{-l})2^{-\frac{k}{2m}}}.
\end{align*} 
For simplicity, we can assume that $\Gamma \geq 1.$
Thus, 
\begin{align*}
    \Vert x_{2lm+k} - \xs \Vert^{\mu^{-2lm-k}} & \leq  \Vert x_k -\xs \Vert^{2^{-\frac{k}{2m}}}\Gamma^{2^{-\frac{k}{2m}}}. \\
    &\leq \Vert x_0 -\xs \Vert^{2^{-\frac{k}{2m}}}\Gamma^{2^{-\frac{k}{2m}}}.
\end{align*} 

Therefore, for $\Vert x_0 -\xs \Vert$ small enough, we get that
for all $k$ such that $0 \leq k <2m$, $\Vert x_0 -\xs \Vert^{2^{-\frac{k}{2m}}}\Gamma^{2^{-\frac{k}{2m}}}<1,$
and hence, $\limsup_s \Vert x_{s} - \xs \Vert^{\mu^{s}} <1. $
From 9.2.7 of \cite{OR}, we then obtain that Broyden's method do have
locally R-order of convergence $2^{\frac{1}{2m}}.$
\end{proof}

\section{Questions on Q-superlinearity}
\label{sec:Qsuperlinearity}
A Q-order of $\mu$ implies an R-order of $\mu.$
The converse is not true.
Over $\R,$ one of the most important
result concerning Broyden's method
is that it is Q-superlinear.
The extension of this result to
the non-archimedean case
remains an open question.

\subsection{Dimension $1$: secant method}

In dimension one, Broyden's method reduces
to the secant method.

It is known since \cite{Bach09}
that the $p$-adic secant method 
applied on polynomials
has order $\Phi,$ the golden ratio.
Its generalization to a general non-archimedean
context is straightforward.

\begin{prop}
Let us assume that $m=1$ and on a neighborhood $U$ of $\xs,$
there is a $c_0 \in \R_{>0}$ such that
$f$ satisfies \eqref{DL2} on $U.$
Then the secant method has locally Q-order of convergence $\Phi.$
\end{prop}
\begin{proof}
Let us assume that we are in the same
context as in the proof of Theorem \ref{thm:R_order},
with some Q-linear convergence of ratio $r <1.$
Let us define $\varepsilon_k= x_k - \xs  $ for $k \in \N.$
For all $k \in \N,$ $\vert \varepsilon_{k+1} \vert  < \vert \varepsilon_k \vert.$
Then by ultrametricity, $ \vert x_{k+1}-x_k \vert =\vert \varepsilon_k \vert$.
Also, we further assume that $c_0 \vert \varepsilon_0 \vert< \vert f'(\xs) \vert$
so that for all $k \in \N,$
 $\vert f'(\xs) \times (x_{k+1}-x_k) \vert >  c_0 \vert (x_{k+1}-x_k) \vert^2,  $
 which also implies by ultrametricity and \eqref{DL2} that for all $k \in \N,$
\[\vert f(x_{k+1})-f(x_k) \vert = \vert f'(\xs) \times (x_{k+1}-x_k) \vert.\]
Similarly,
$ \vert f(x_k) \vert = \vert f'(\xs ) \vert \vert \varepsilon_k \vert.$

Now, let $n \in \Z_{>0}.$
Broyden's iteration is given by:
\[x_{n+1}=x_n- \frac{x_n-x_{n-1}}{f(x_n) -f(x_{n-1})}. \]
It rewrites as:
\begin{align*}
    \vert \varepsilon_{n+1} \vert &= \vert \varepsilon_n - \frac{\varepsilon_n f(x_n)- \varepsilon_{n-1} f(x_n)}{f(x_n)-f(x_{n-1})} \vert 
    = \vert \frac{\varepsilon_{n-1} f(x_n)- \varepsilon_{n} f(x_{n-1})}{f(x_n)-f(x_{n-1})} \vert \\
    &\leq c_0\frac{\max \left( \vert \varepsilon_{n-1} \vert \vert \varepsilon_n \vert^2 \ , \  \vert \varepsilon_{n-1} \vert^2 \vert \varepsilon_n \vert \right)}{\vert f(x_n)-f(x_{n-1}) \vert} 
     \leq \frac{c_0}{\vert f'(\xs ) \vert} \vert \varepsilon_n \vert \vert \varepsilon_{n-1} \vert.
\end{align*}
Let us write $C=\frac{c_0}{\vert f'(\xs ) \vert} $
and $v_n= C \varepsilon_n.$
Then, $v_{n+1} \leq v_n v_{n-1}$ for any $n>0$
and consequently,
\[\frac{v_{n+1}}{v_n^\Phi} \leq v_n^{1-\Phi} v_{n-1} \leq \left( \frac{v_n}{v_{n-1}^\Phi} \right)^{1-\Phi}, \]
as $\Phi^2=\Phi+1.$
If we define $(Y_n)_{n \in \Z_{\geq 1}}$ by $Y_1=\frac{v_1}{v_0^\Phi}$
and $Y_{n+1}=Y_n^{1-\Phi},$ then $\frac{v_{n+1}}{v_n^\Phi} \leq Y_n.$
Since $\vert 1-\Phi \vert <1,$ then $Y_n$ converges to $1.$
Therefore, it is bounded by some $D \in \R_+,$
and $\frac{v_{n+1}}{v_n^\Phi} \leq D$ for all $n \in \Z_{\geq 1}.$
This concludes the proof.
\end{proof}

\subsection{General case}

Over $\R,$ Broyden's method is known to converge
Q-superlinearly.
The key point is that for any 
$E \in M_m(\R)$ and $s \in \R^m \setminus \lbrace 0 \rbrace,$

\begin{equation}
\Vert E \left( I- \frac{s\cdot s\tran}{(s\tran \cdot s)} \right) \Vert_F^2= \Vert E \Vert_F^2-\left( \frac{\Vert Es \Vert_2}{\Vert s \Vert_2} \right)^2, \label{eqn:eq_Frob_sur_R}
\end{equation}
equation $(5.5)$ of \cite{DM77}.
The minus sign is a blessing as it allows
the appearance of  a telescopic sum which
plays a key role in proving 
that $\frac{\Vert x_{n+1}- \xs \Vert}{\Vert x_n - \xs \Vert}$
converges to zero.
Unfortunately, there does not seem to be
a non-archimedean analogue to this
equality.
Thanks to Theorem \ref{thm:R_order},
we nevertheless believe in the following conjecture.

\begin{conj}
In the same setting as Theorem \ref{thm:R_order},
Broyden's method
has locally Q-superlinear convergence. \label{conj1}
\end{conj}

\section{Finite precision}
\label{sec:precision_finie}

\subsection{Design and notations}

One remarkable feature of Newton's method in
an ultrametric context is the way it can
handle precision. 
For example, if $\pi$ is a uniformizer,
if we assume that $\Vert f'(\xs)^{-1} \Vert =1,$
$x_n$ known at precision $O(\pi^{2^n})$
is enough to obtain $x_{n+1}$ at
precision $O(\pi^{2^{n+1}}).$
To that intent, it thus suffices to double
the precision at each new iteration.
Hence  the working precision of Newton's method
can be taken to grow at the same rate as the rate of convergence.

The handling of precision is more subtle in Broyden.
This is however crucial to design efficient implementations.
Note that in the real numerical setting, most works using Broyden's methods
are employing fixed finite precision arithmetic, and do not address precision.
Additionally, the lack of a knowledge of a precise exponent of convergence requires special care, and the presence of a division
also complicates the matter.
We explain hereafter how to cope with those issues.


For simplicity, we will make the following hypotheses
throughout this section, which correspond
to the standard ones in the Newton-Hensel method.
They are that the  starting $x_0$ and $B_0$
are in a basin of convergence at least linear.
This allows us to replace any encountered
$x_n$ by its lift $\tilde{x}_n$ 
to a higher precision
(and same for $B_n$).
Indeed, $\tilde{x}_n$ will still be in the basin
of convergence and then follows
the same convergence property.
These liftings allow to mitigate the fact that
some divisions are reducing the amount
of precision so that only arbitrary added digits
are destroyed by the divisions.\footnote{This an example
of an adaptive method, which can also be used
in Newton's method when divisions occur.}

\begin{ass}
We assume that $x_0$ and $\xs$ are in $\OK$,
and that $\Vert f'(\xs) \Vert =\Vert f'(\xs)^{-1} \Vert =\Vert B_0 \Vert = \Vert B_0^{-1} \Vert=1.$
We also assume that some $\rho_1 \leq 1$ and $\rho_2 \leq 1$ are given such
that $ B(\xs, \rho_1) \times B(f'(\xs), \rho_2),$
is a basin of convergence at least linear and for any
$x \in B(\xs, \rho_1)$, and $\rho \leq \rho_1,$ $f(x+B(0, \rho))=f(x)+f'(\xs) \cdot B(0, \rho)$ (see the Precision Lemma 3.16 of \cite{caruso2017})
\label{Ass:simple_precision}
\end{ass}
The assumption on $B_0$ and $f'(\xs)$
states that they are unimodular,
which is the best one can assume regarding to conditioning
and precision. Indeed if $M \in GL_m(K)$ is unimodular ($\Vert M \Vert = \Vert M^{-1} \Vert=1$),
then for any $x \in K^m,$ $\Vert M x \Vert=\Vert x \Vert.$
Over $\Qp,$ $M \in M_m(\Zp)$ is unimodular if and only if
its reduction in $M_m(\Z/ p \Z)$ is invertible 
(and \textit{idem} for $\Qt$ and $\Q$).
The last assumption is there to provide the precision on
the evaluations $f(x_k)$'s. It is satisfied if $f \in \OK [X_1,\dots,X_m].$

\noindent {\bf Precision and complexity settings.}
Let ${\sf M}(N)$ be a superadditive upper-bound 
on the 
arithmetic complexity over the residue field of $\OK$
for the computation of the product of two elements in $\OK$
at precision $O(\pi^N)$,
and $L$ be the size of a straight-line program that computes the system $f$.
One can take ${\sf M}(N)\in O\tilde(N)$.

Working over $K$ with {\em zealous} arithmetic, the ultrametric counterpart
of interval arithmetic \cite[\S~2.1]{caruso2017},
the interval of integers
$\openr{a,\ b}$ indicates the coefficients of an element $x\in K$
represented in the computer as $x=\sum_{i=a}^{b-1} x_i \pi^i$, with $x_i\in \OK /\langle \pi \rangle$.
In this way $\val(x) = a$, its {\em absolute precision} is  $\abs(x)=b$, and
its {\em relative precision} is $\rel(x)=b-a$.
We recall the usual precision formulae, and assume in the algorithm below
that it is how the software manages zealous arithmetic (as in Magma, SageMath, Pari). See \textit{loc. cit.} for more details.
\begin{align*}
\openr{a,\ b} \times \openr{c,\ d} &= \openr{a+c,\ \min (a+d,\ b+c)} \\
\openr{a,\ b} / \openr{c,\ d} &= \openr{a-c,\ \min (a+d-2c,\ b-c)} \tag{P}\label{tag:P}
\end{align*}
The cost of multiplying  two elements of relative precision $a$ and $b$ is within $\M(\max(a,\ b))$, and to divide one by the other is in $ 4 \M (\max(a,\ b))+\max(a, \ b )$ ~\cite[Thm~9.4]{GaGe03}.

To perform changes in the precision,
we use the same notation as Magma's function for doing so.
If $x$ has interval $\openr{a,\ b}$, the (destructive) procedure
``ChangePrec(\textasciitilde $x,\ c$)'' 
either truncates $x$  to absolute precision $c$ if $c\le b$, or
lifts with zero coefficients $0 \pi^b+\cdots+0 \pi^{c-1}$
to fit the interval $\openr{a,\ c}$, if $c > b$.
The non-destructive
counterpart is denoted ``ChangePrec($x, \ c$)'' without \textasciitilde.

\subsection{Effective Broyden's method}
We start from an initial 
approximation $x_0$ at precision one, for example
given by a modular method.
The inverse of the Jacobian at precision one
provides $B_0^{-1}$. It yields a cost of $O(m^\omega)$,
but the complexity analysis of Remark~\ref{rmk:NvsB}
shows that it is negligible.
Obtaining these data is not always  obvious~\cite{FY80},
but is the standard hypothesis in the context of modular methods.
We  write $v_k=\val(f_k)$,
\medskip

\noindent {\bf In an ideal situation.}
Assume an oracle provides the valuations $v_0, v_1, v_2, \dots, v_n,\dots $
(computed by a Broyden method at arbitrarily large precision).
From this ideal situation, we derive the simple and costless modifications
required in reality. This analysis allows us
to know how efficient can a Broyden method be, which is noteworthy for
comparing it to Newton's.
The implementation of Iteration $n$ ($n=0$ included) follows the lines hereunder.
The rightmost interval indicates the output interval precision
of the object computed (following~\eqref{tag:P}),
while the middle indicates a complexity estimate.
\medskip

\noindent\underline{Input:}
\begin{enumerate*}
\item $B_n^{-1}$ has interval $\openr{ 0,\ v_n}$ and is unimodular.\\
\item {\small $x_n$ has interval $\openr{0,\ v_n+v_{n+1}}$ (non-zero entries in $\openr{0,\ v_{n-1}+v_{n}}$).}\\
\item $f_n$ has interval $\openr{v_n,\ v_n +v_{n+1}}$.
\end{enumerate*}

\medskip

\noindent\underline{Output:} \begin{enumerate*}[label=(\roman*)]
\item $B_{n+1}^{-1}$ with interval $\openr{0,\ v_{n+1}}$, 
($\val (\det (B_n^{-1}))=0$).\\
\item {\small $x_{n+1}$ in the interval $\openr{0,\ v_{n+1}+v_{n+2}}$ (non-zero entries in $\openr{0,\ v_n + v_{n+1}}$).}\\
\item $f_{n+1}$ in the interval $\openr{0,\ v_{n+1}+v_{n+2}}$.
\end{enumerate*}

\medskip

\begin{enumerate}
\item\label{enum:un} ChangePrec(\textasciitilde$B_n^{-1}, \ v_{n+1}$) ; \hfill   \qquad $\openr{0,\ v_{n+1}}$ 
\item\label{enum:sn}  $s_n \leftarrow - B_n^{-1} \cdot f_n$ ; \hfill $m^2 \M(v_{n+1})$ \qquad $\openr{v_n,\ v_n+v_{n+1}}$ 
\item\label{enum:xn+1} $x_{n+1} \leftarrow x_n + s_n$ ;\hfill $\openr{0,\ v_n+v_{n+1}}$
  \item ChangePrec(\textasciitilde$x_{n+1}, \ v_{n+1}+v_{n+2}$) ; \hfill $\openr{0,\ v_{n+1}+v_{n+2}}$
  \item\label{enum:fn_plus_un}  $f_{n+1}\leftarrow f(x_{n+1})$ ;  \\
    {\tiny .} \hfill $ L \cdot \M(v_{n+1}+v_{n+2}) $ \qquad   $\openr{v_{n+1},\ v_{n+1}+v_{n+2}}$
    \item\label{enum:valF+1} $\wb{f_{n+1}} \leftarrow$ ChangePrec($f_{n+1}, \ v_{n} + v_{n+1}$) ; \hfill $\openr{v_{n+1},\ v_{n+1} + v_n}$
    \item $h_n \leftarrow B_n^{-1} \cdot \wb{f_{n+1}}$ ;  \hfill $m^2 \M(v_{n+1})$ \qquad $\openr{v_{n+1},\ v_n + v_{n+1}}$
    \item $u_n \leftarrow $  Eq.\eqref{eqn:def_simple_un} ; \hfill   (negligible) \qquad   $\openr{-v_{n}, \  v_{n+1}-v_{n}}$
    \item $r_n \leftarrow u_{n}^T \cdot $ ChangePrec($B_{n}^{-1}, \ v_n$) ; \hfill  {\small    $m^2 \M(v_{n+1})$ \quad 
  $\openr{ -v_n  ,\ 0 }$   }
      \item ChangePrec(\textasciitilde $\wb{f_{n+1}},\ 2 v_n$) ; \hfill $\openr{v_{n+1},\ 2 v_n}$
      \item\label{enum:toDiv} ${\rm den} \leftarrow 1+r_n \cdot \wb{f_{n+1}}$ ;
        \hfill $ m \M(v_{n+1})$ \qquad $\openr{0,\  v_{n} }$
    \item\label{enum:num} ${\rm Num}\leftarrow h_n \cdot r_n$ ; \hfill $ m^2 \M(v_{n})$ \qquad $\openr{ v_{n+1}-v_n,\ v_{n+1} }$
    \item\label{enum:div} ${\rm N}_n \leftarrow {\rm Num}/{\rm den}$ ; \hfill $ 4 m^2 \M(v_{n})$ \qquad $\openr{v_{n+1}-v_n ,\ v_{n+1}}$
    \item $B_{n+1}^{-1} \leftarrow B_{n}^{-1} - {\rm N}_n$ ; \hfill $\openr{ 0,\ v_{n+1}}$
      \item\label{enum:lastB} return $B_{n+1}^{-1},\ x_{n+1},\ f_{n+1}$
\end{enumerate}

We emphasize again that thanks to the careful changes of precision
undertaken, the precisions are automatically managed
by the software, would it have zealous arithmetic implemented.
It is then immediate to check that the output verifies the specifications.
Moreover from the positive valuation of ${\rm N}_n$ it is clear
that $B_{n+1}$ is unimodular. Thus Iteration~$n+1$
can be initiated with these outputs.

\medskip

\noindent {\bf Complexity of the ideal situation.}
The arithmetic cost of Iteration $n$ is within $ (3 m^2 +m) \M(v_{n+1}) +  5 m^2 \M(v_{n}) + L\cdot  \M(v_{n+2} + v_{n+1}) $.
If we assume an exponent of convergence $\alpha>1$, \textit{i.e.}
$v_{n+1} \approx \alpha v_n$ for ``not too small'' $n$, then
the total cost to reach a precision $N\approx \alpha^{\ell+1}\approx v_{\ell+1}$
($\ell$ steps, including a $0$-th one) is upper-bounded by
\begin{equation}
  (  5 m^2 + (3 m^2 +m )\alpha^2+ L (1+\alpha)^2 \alpha^2)
\M( N/(\alpha-1))  \label{tag:I}
\end{equation}

\noindent {\bf In reality.}
Using the same notations and inputs at Iteration~$n$ as in the ideal situation above,
what changes in reality is that while $v_n$ is known $v_{n+1}$ and $v_{n+2}$ are not,
but are approximated by $\alpha v_n \ge v_{n+1}$ and $\alpha^2 v_n\ge v_{n+2}$
respectively, where $\alpha$ is fixed by the user.
Precisely,  $B_n^{-1}$ and $x_{n}$ are known at the correct
precision, but $f_n$
has an approximated interval $\openr{0,\ v_{n}+\alpha v_{n}}$.
To minimize the overhead cost it induces compared to the ideal situation,
once we know $v_{n+1}$ (Line~\ref{enum:fn_plus_un})
we insert some intermediate corrective steps denoted (5.1)-(5.5) thereafter,
between Line~\eqref{enum:fn_plus_un} and Line~\eqref{enum:valF+1};
they require no arithmetic operations.


\smallskip

(5.1) ChangePrec(\textasciitilde$B_n^{-1},\ v_{n+1}$)

(5.2) ChangePrec(\textasciitilde$s_n,\ v_n+v_{n+1}$)

(5.3) Tune $\alpha$ if necessary using the new ratio $\frac{v_{n+1}}{v_n}$

(5.4) ChangePrec(\textasciitilde$x_n,\ v_{n+1} + \alpha v_{n+1}$)

(5.5) ChangePrec(\textasciitilde$f_{n+1},\ v_{n+1} + \alpha v_{n+1}$)

\smallskip

Most importantly,
the remaining Lines~\eqref{enum:valF+1}-\eqref{enum:lastB}
are not impacted since these computations involve now the known $v_{n+1}$ (and not
the unknown $v_{n+2}$): the intervals, and thus costs obtained
are the same as in the ideal situation.
On the other hand, Lines~\eqref{enum:un}-\eqref{enum:fn_plus_un} are performed as such with an overhead
cost. Among them, only Lines~\eqref{enum:sn}, 
\eqref{enum:fn_plus_un} have a non negligible cost.
At Line~\eqref{enum:sn},  $B_n^{-1}$ has approximated interval $\openr{0,\ \alpha v_n}$,
yielding a cost of $m^2 \M(\alpha v_n)$.
At Line~\eqref{enum:fn_plus_un} $x_{n+1}$ has approximated interval
$\openr{0,\ v_n(\alpha+\alpha^2)}$, yielding a cost of $L \M(v_n (\alpha (1+\alpha)))$.
Thus the overhead cost ``ovh${}_n$'' at Iteration~$n$  is:
\begin{equation}\label{eqn:ovh}
       m^2 (\M(\alpha v_n) - \M(v_{n+1})) +
      L (\M (v_n \alpha(1 +
    \alpha)) - \M(v_{n+1}+v_{n+2}) )
\end{equation}
    This quantity depends on  the gaps $ \alpha v_n - v_{n+1} $ and
    $\alpha^2 v_n - v_{n+2}$.
    These gaps increase with $n$,
    but, thanks to the tuning of Step $(5.3)$, reasonably at a linear rate:
        \vspace{-0.75mm}
        \begin{ass}\label{ass:gap}
          The ``error gap''  $\ \ \ \vert \alpha v_n - v_{n+1} \vert  = O ( n)$.
        \end{ass}
        \vspace{-0.75mm}
        Under this assumption it is easy to (crudely)
        bound   $\sum\nolimits_{n=0}^{\ell+1} {\rm ovh}_n$  of Eq.~\eqref{eqn:ovh} by
        $(L + m^2) O(N \log(N))$. Being independent on $\alpha$   this is negligible
        in front of  $O(L + m^2)  \M(\frac{N}{\alpha-1}) $
        for $\alpha<2$. The theorem below wraps up  the considerations made above
        with Eq.~\eqref{tag:I}:
\begin{theo}\label{th:prec}
If Broyden's method has Q-order of convergence $\alpha$ on $B(\xs, \rho_1) \times B(f'(\xs), \rho_2)$,
then under Assumption \ref{Ass:simple_precision} and~\ref{ass:gap}, the cost of
computing $\xs + O(\pi^N)$ is in $O\left((m^2 +L)\right) {\sf M}\left( \frac{N }{\alpha-1}\right)$.
\end{theo}
\begin{rmk}\label{rmk:NvsB}
  Understanding the $Q$-order of convergence is a major and notoriously difficult
  problem in the numerical analysis community. Numerical evidence shows it deteriorates with $m$,
  and is larger than $2^{1/2m}$ (Theorems~\ref{thm:R_order}-\ref{th:bigone}). Some experiments suggest that
  taking $\alpha \approx 2^{1/m}$ is not unreasonable.
  We then
  get a cost in $O\left((m^2 +L){\sf M}\left( \frac{N }{\alpha-1}\right)\right) \approx  O\left((m^2 +L){\sf M}\left( Nm\right)\right)$.
For comparison, denoting $\omega<3$ the exponent
of the cost of matrix product, the standard analysis
of Newton's method for rational fractions would lead
to $O\left((m^\omega + m L){\sf M}\left(N\right)\right)$.
Consequently, in this setting, for large $m$, there is little hope that
Broyden's method can outperform Newton's when both are available.
Remember though other worthwile
applications in the paragraph 
{\em ``Motivations''} in Introduction.
\end{rmk}
\section{Numerical data}
\label{sec:data}

An implementation of our ultrametric Broyden method
in Magma \cite{magma} with more data is available at \url{http://xdahan.sakura.ne.jp/broyden20.html}.
We report the data obtained  using the three families of
systems, derived from
page 36 of \cite{Lecerf2001:phd}. The families are
indexed by $t \in \pi \OK$:
{\small
\begin{itemize}[leftmargin=*]
\item $F_1= \big( (x_1-1)^2 + (x_2-1)^2 - 4 -t x_1x_2 -t^2x_1 ,
(x_1 +1 )^2 +(x_2 +1)^2 -4 -t x_1 \big)$ in $K[x_1,x_2]$.
\item $F_2= \big( (x_1-1)^2 + (x_2-1)^2 +(x_3-1)^2 - 5 -t -t^2 ,
(x_1 +1 )^2 +(x_2 +1)^2 +(x_3+1)^2 -5 -t,
2x_1^2 + x_2^2 + x_3^2 -3 - t^2 \big)$ in $K[x_1,x_2, x_3]$. 
\item $F_3= \big( (x_1-1)^2 + (x_2-1)^2 +(x_3-1)^2 +(x_4-1)^2 - 8 -t -t^2 ,
(x_1 +1 )^2 + (x_2 +1)^2 + (x_3 + 1)^2 + (x_4 + 1 )^2 - 8 - t ,
2\,x_1^2 + x_2^2 + x_3^2 + x_4^2 -5 - t^2 ,
2\,x_1\,x_2 + x_3\,x_2 - 2\, x_3\,x_4 + 2\,x_4\,x_1 +3 - t^2 \big)$ in $K[x_1,x_2, x_3,x_4]$.
\end{itemize}
}
Valuation of $f(x_k)$ and numerical estimation of the order of Q-convergence for $\Qt$ are
compiled in the following graphic. For $K=\Qp,$ and $\Fp \llbracket t \rrbracket$ with $p=17$
we experienced the same behaviour.
\begin{center}
  \includegraphics[height=3.4cm]{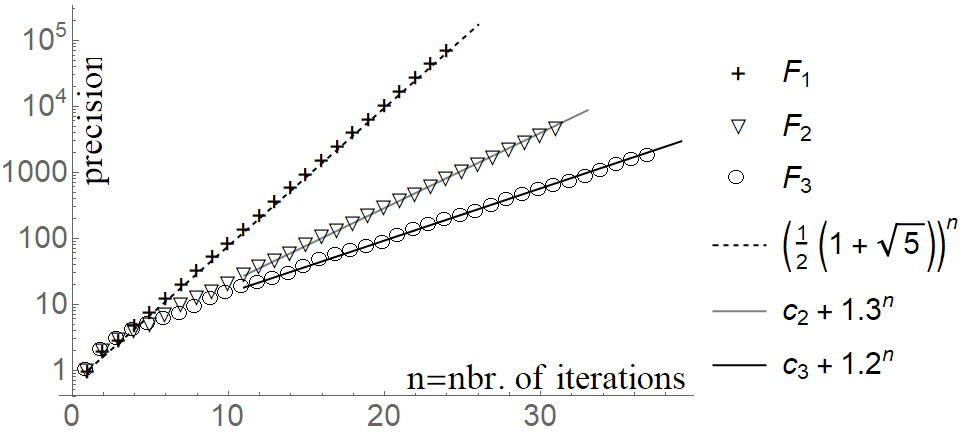}
\end{center}
\vspace{-5mm}

  \bibliographystyle{plain}
\bibliography{biblio}

\end{document}